\documentclass[11pt]{article}
\usepackage[left=1in,top=1in,right=1in,bottom=1in]{geometry}
\usepackage{times}

\usepackage{graphicx}
\usepackage{url}
\usepackage{verbatim}

\usepackage{amsmath}
\usepackage{amssymb}
\usepackage{amsthm}
\usepackage{algorithm}

\newtheorem{definition}{Definition}

\newtheorem{proposition}{Proposition}

\begin{document}

\title{Mistakes in Games}

\author{Sam Ganzfried\\
Ganzfried Research \\
sam@ganzfriedresearch.com
}


\date{\vspace{-5ex}}

\maketitle

\begin{abstract}
We define a new concept of ``mistake'' strategies and actions for strategic-form and extensive-form games, analyze the relationship to prior main game-theoretic solution concepts, study algorithms for computation, and explore practicality. This concept has potential applications to cybersecurity, for example detecting whether a human player is illegally using real-time assistance in games like poker.
\end{abstract}

\section{Introduction}
\label{se:intro}
A strategy in a game can obtain poor performance for one of two primary reasons. The first is that it can put positive probability on a specific action (or pure strategy for the case of a strategic-form game) that is clearly ``bad'' and should never be played. And the second is that the strategy does not assign probabilities for the actions that properly balance and lead to the strategy being ``predictable'' and exploitable, despite not playing any action that is clearly bad on its own.

The classic concept used to describe these ``bad'' actions and strategies is that of dominance. A pure strategy in a strategic-form game is strictly dominated if it performs strictly worse than another strategy for all possible strategies that can be used by the opponent. Let $S_i$ denote the set of pure strategies of player $i$, $S_{-i}$ denote the strategy space of all players excluding $i$, and $u_i$ denote the utility function for player $i$. Then pure strategy $s'_i$ \emph{strictly dominates} $s''_i$ if $u_i(s'_i,s_{-i}) > u_i(s''_i,s_{-i})$ for all $s_{-i} \in S_{-i}.$ 

Playing any strategy that is strictly dominated is clearly bad (since there exists another strategy that will always do strictly better). No rational agent would ever play a strategy that is strictly dominated. Often strictly-dominated strategies are removed from the game in advance of game-theoretic analysis (e.g., computing a solution concept such as Nash equilibrium). It is clear that no strategy that is strictly dominated can receive any probability mass in any Nash equilibrium. In fact we can go further, and iteratively remove first one strictly dominated strategy for one player, then another strictly-dominated strategy for a player, etc., until there are no more strictly-dominated strategies remaining for any player. It is well known both that this process results in the same reduced strategy spaces of the players regardless of the order of elimination of the dominated strategies, and also that the set of Nash equilibria in the reduced game is identical to the set of equilibria in the initial game. Furthermore, this procedure is easy to apply computationally. For these reasons, this procedure is often used as a pre-processing step before applying an algorithm for computing Nash equilibrium.

A further solution concept that has also been widely studied is that of \emph{weak dominance}. The pure strategy $s'_i$ \emph{weakly dominates} $s''_i$ if $u_i(s'_i,s_{-i}) >= u_i(s''_i,s_{-i})$ for all $s_{-i} \in S_{-i},$ where the inequality is strict for at least one of the $s_{-i}$. (Note that the latter condition is included in case the strategies $s'_i$ and $s''_i$ obtain identical payoffs against all opponents, in which case it does not seem right to say that one dominates the other.) As for strict dominance, it also seems clearly bad and irrational for an agent to play a strategy that is weakly dominated, as another strategy will always do at least as well and sometimes strictly better. Unlike for strict dominance, it is possible for a weakly-dominated strategy to be played with positive probability in equilibrium, and if weakly-dominated strategies are removed iteratively for the players the order of removal can result in reduced games that have different sets of equilibria (though all of the reduced games will contain at least one Nash equilibrium from the initial game). The procedure of iteratively removing weakly-dominated strategies can also be done efficiently, and has also been applied as a pre-processing step for algorithms with the goal of computing one Nash equilibrium in large games.

Note that these approaches can both be extended naturally to more than two players, and apply to both zero-sum and non-zero-sum games, though in this paper we will be primarily studying two-player zero-sum games. Pure strategies can also be dominated (strictly or weakly) by mixed strategies in addition to just pure strategies, and there exists an efficient algorithm based on a linear program formulation for determining whether a strategy is dominated by any mixed strategy~\cite{Conitzer05:Complexity_of_iterated_dominance}. 

Several concepts have been considered that extend the concept of dominance in several ways, though none have received significant attention or been applied in practice. One work considers a generalized strategy eliminability criterion for two player general-sum games that considers whether a given strategy is eliminable relative to given sets of dominator and eliminee subsets of the players' strategies~\cite{Conitzer05:Generalized}. The concept spans a spectrum of eliminability criteria from strict dominance (when the sets are as small as possible) to Nash equilibrium (when the sets are as large as possible). However, the problem of eliminating a strategy according to this criterion in general is computationally hard (coNP-complete), though in certain special cases can be done efficiently.

 
For extensive-form games of imperfect information (which are often represented as game trees, with chance moves, player moves, information sets of game states, and payoffs at leaf nodes), many of these approaches are not applicable (even for two-player zero-sum). The concept of domination is still applicable, but in a different way. In extensive-form games, a pure strategy consists of an entire contingency plan, which specifies which action should be taken at every possible information set that could be encountered (even for those that are impossible to reach given actions taken previously). There are exponentially many pure strategies in the size of the game tree. If a pure strategy is dominated in an extensive-form game, we cannot simply remove it from the game tree, because the strategy consists of a selection of actions at all information sets, and some of the actions may also be part of other non-dominated strategies. One option would be to convert the game to a strategic-form game, and then we can eliminate dominated strategies in the strategic-form game as described above; however, the size of the strategic-form game is exponential in the size of the extensive-form game, and performing the transformation is often computationally infeasible.

In order to apply domination in a computationally feasible way in extensive-form games, we need to apply a significantly stronger form. Rather than eliminating dominated \emph{strategies}, which constitute full contingency plans, instead we can eliminate dominated \emph{actions}---a specific action at a specific information set. In order for one action $a_i$ at an information set $I$ to dominate action $a'_i$ also at I, it must be the case that the payoff for player $i$ at every possible leaf node succeeding $a_i$ exceeds the payoff for player $i$ at every possible leaf node succeeding $a'_i$. This is an extremely strong requirement, and it is very difficult for actions to be dominated according to this definition (particularly for actions closer to the root of the tree). Note that for dominated actions, we can apply analogous concepts of strict, weak, and iterated domination, as described above for strategic-form games. These can still be applied efficiently, however are less likely to be useful due to the stricter requirements. 

As an example, consider the small game of Kuhn poker, with game tree shown in Figure~\ref{fi:kuhn-poker}. In Kuhn poker, there is a three-card deck consisting of a King (K), Queen (Q), and Jack (J). Initially each of two players places an ante of \$1 (so there is an initial \emph{pot} of \$2). Each player is dealt one of the cards privately (while the third card is not dealt to any player). Player 1 is allowed to bet \$1 or to \emph{check}. If P1 bets, then P2 is allowed to \emph{call} and match the bet, or to \emph{fold} and forfeit the hand. If player 2 matches the bet, then the player with the winning hand wins the full pot of \$4. If player 2 folds, then player 1 wins the pot (of \$2). If player 1 chose to check in the first round, then player 2 is allowed to bet \$1 or to check. If player 2 also checks, then the player with the highest card wins the pot of \$2. If player 2 bets when facing a check, then player 1 can call or fold.   

\begin{figure}[!ht]
\centering
\includegraphics[scale=0.8]{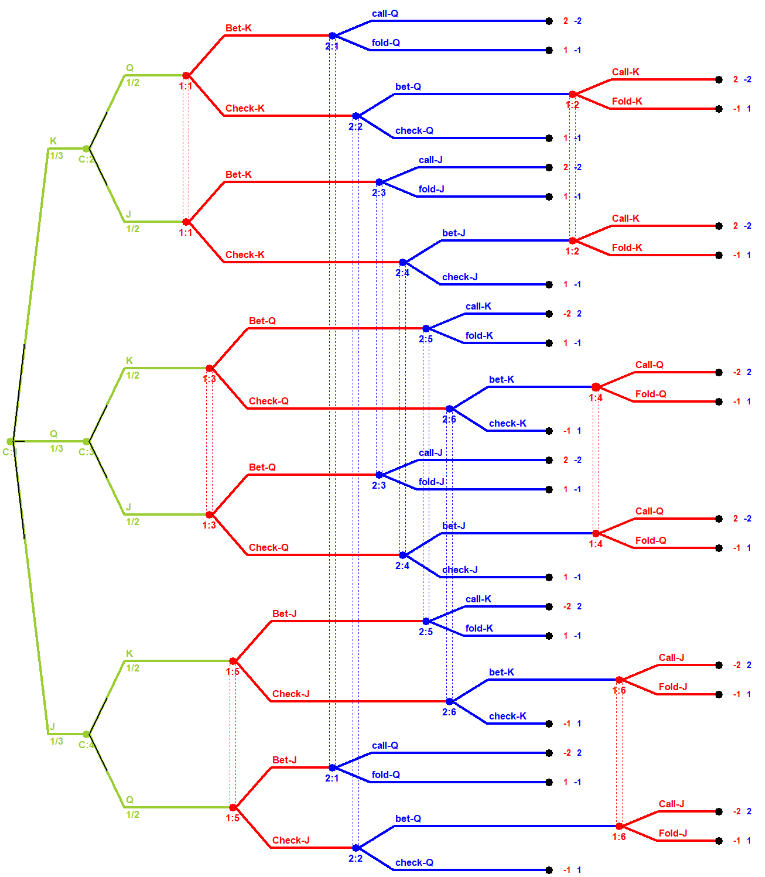}
\caption{Original game tree for Kuhn poker.}
\label{fi:kuhn-poker}
\end{figure}

In Kuhn poker, it turns out that domination plays a very significant role. First, we can remove the actions for both players where they call a bet with a Jack (because calling will always result in losing an additional \$1 regardless of the other player's card). Similarly, we can also remove the actions for both players where they fold with a King when facing a bet. For player 2, we can also remove the action of checking with King after facing a check by player 1, because betting will always result in at least as big a payoff. Note that for player 1 we cannot remove the action of checking with a King as the initial action, because there could be benefit to checking and inducing player 2 to bet with a weaker hand that would have folded had we bet ourselves.
As it turns out, after removing these dominated actions from the tree, several further actions become dominated. Once neither player is calling a bet with a Jack or folding to a bet with a King, it turns out that it is now dominated for either player to bet with a Queen. The reduced game after removing all iteratively dominated actions is shown in Figure~\ref{fi:kuhn-poker-IDR}. 

\begin{figure}[!ht]
\centering
\includegraphics[scale=0.7]{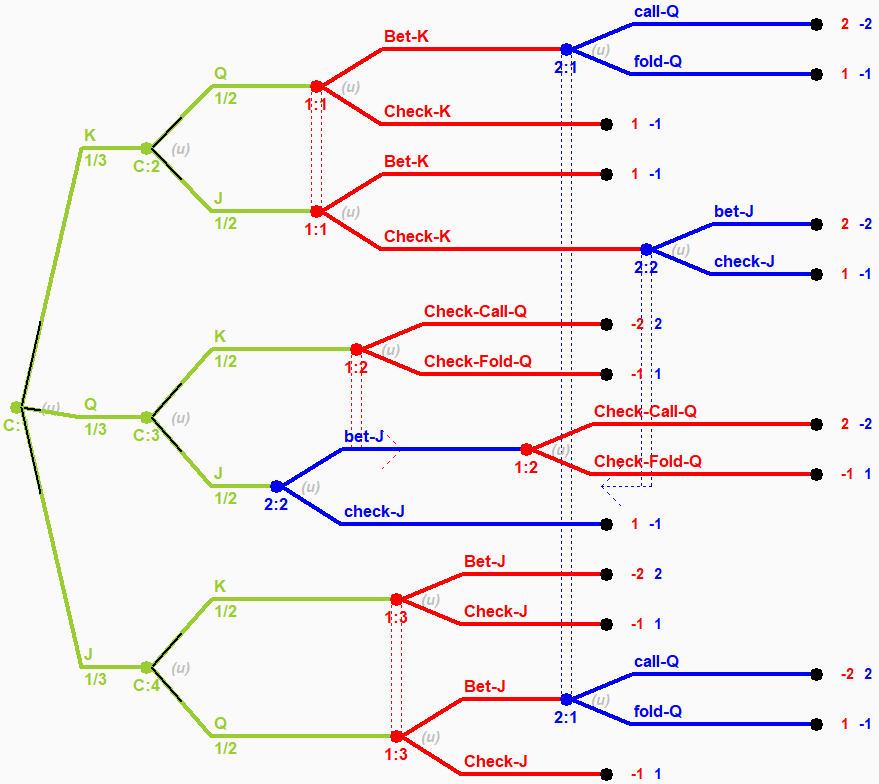}
\caption{Game tree for Kuhn poker after iterated removal of dominated actions.}
\label{fi:kuhn-poker-IDR}
\end{figure}

As it turns out, in this game simply avoiding dominated actions has been demonstrated to obtain reasonably strong performance~\cite{Gibson14:Regret}. A tournament was run between the following six agents: uniform random strategy (Uni), a strategy that plays no strictly dominated actions (does not call with Jack, fold King, or check King when facing a check) but is otherwise uniform random (ND), a strategy that plays no iteratively dominated actions (no dominated actions and neither player bets with Queen) but is otherwise uniform random (NID), and three Nash equilibrium strategies (note that the game contains infinitely many equilibrium strategies parametrized by a single parameter, which have been previously computed~\cite{Kuhn50:Simplified}). The Nash equilibrium strategies outperform the others (both in terms of head-to-head experiments and worst-case exploitability), followed by NID and then ND, with Uni performing worst (unsurprisingly). Despite losing to the equilibrium, both NID and ND perform significantly better than Uni, and obtain positive utility overall taking into account all opponents. This demonstrates that, for this game, simply avoiding (iteratively) dominated actions can lead to a significant improvement in performance.

However, as we consider even slightly larger games, the role of dominated actions quickly becomes irrelevant, as it becomes so difficult for an action to dominate another that very few actions are dominated. For example, if we consider the same game of Kuhn poker but with a 4-card deck instead of a 3-card deck, the only actions that are iteratively dominated are those analogous to the first round of dominated strategies in classic Kuhn poker (calling with the worst possible hand, folding with the best possible hand, or checking with the best possible hand when facing a check). We can not rule out potentially calling with the second worst possible hand, or folding with the second best possible hand. And furthermore, we cannot iteratively rule out any additional actions after the first round of eliminations, in contrast to the 3-card variant where we can remove the aggressive betting actions with Queen. It turns out that this same conclusion will extend to an $n$-card variant; regardless of the deck size (given that it exceeds 3), the only dominated actions will be to call with the worst possible hand, fold with the best possible hand, and check the best possible hand when facing a check, and no additional actions will be iteratively dominated after removing these. While removing these three actions from the game tree in the small 3-card variant ends up being significant, this comprises such a tiny fraction of the overall game tree for larger $n$ that it is meaningless. For even larger variants, such as no-limit Texas hold 'em, the conclusion is similar. The only actions that are (iteratively) dominated will be terminal actions of calling a bet with the worst possible hand, folding to a bet with the best possible hand, or checking after the opponent has checked with the best possible hand. 
The requirement that an action must outperform another action at every possible leaf node that can result from taking the action is such a strong requirement that the concept of dominated actions is essentially meaningless in most realistic games other than very small games or games that have been specifically contrived.

Nonetheless, many actions in large imperfect-information games, such as poker, seem to be clear ``mistakes,'' such as calling with a weak hand or folding with a reasonably strong hand on any round. Identifying such mistakes could be of value for several reasons. First, if there are a lot of them, then identifying and removing them from the game tree could lead to significantly faster running times of equilibrium-finding algorithms. Furthermore, it could be helpful for human players, who are looking to improve their play but not capable of computing a close approximation of full Nash equilibrium strategies, to at least be able to assess their play by determining whether specific actions are mistakes on their own.  

While no more meaningful concepts of ``mistake'' actions have been proposed for extensive-form games, the concept of removing bad actions has been taken into account within the operation of one of the popular algorithms for equilibrium computation. The counterfactual regret minimization algorithm (CFR)~\cite{Zinkevich07:Regret} has been applied to approximate Nash equilibrium strategies in large imperfect-information games, and has led to computer agents that defeat some of the strongest human players~\cite{Moravcik17b:DeepStack,Brown17:Superhuman,Brown19:Superhuman}. 
Recently it has been shown that the speed of convergence can be significantly improved if actions that are performing poorly within the course of the algorithm (i.e., have negative or small positive regret for not playing the action) are pruned early on in the algorithm~\cite{Brown17b:Reduced,Brown15:Regret,Brown17:Dynamic}. These algorithms periodically check for actions that have regret below a certain threshold (which may be modified dynamically) and remove them from the game tree so that they are no longer updated in future iterations. Note that the approach is particularly beneficial when bad actions are removed early on on the game tree, as the entire subtree following that action can then also be removed. (Recall that for the concept of dominated actions it is much easier for an action to be dominated later on in the game tree than early because there are fewer succeeding leaf nodes for which worse payoff is needed.) 
While these approaches lead to improved runtimes for equilibrium computation, they do not lead to a clear definition of ``mistake'' actions. 

CFR can also be applied to games with more than two agents and non-zero-sum games, though is not guaranteed to converge to Nash equilibrium for these games. However, it does have one notable property that it is guaranteed in the limit to avoid putting positive probability on strictly dominated strategies or actions~\cite{Gibson14:Regret}. While this guarantee is not useful in practice, it is the only theoretical result known so far for these game classes.

Considering all the prior results, there fails to be a rigorous concept of what constitutes a mistake action in two-player zero-sum extensive-form games beyond dominance, which is useless in games other than certain very small games or games that are specifically contrived. 

\section{Mistakes}
\label{se:mistakes}
We define a \emph{mistake} as, simply, a strategy that is played with probability zero in every Nash equilibrium. This definition can apply to both strategic-form and extensive-form games, and we can also define an action in an extensive-form game to be a mistake if it is played with probability zero in every Nash equilibrium. It has been proven that even for two-player general-sum symmetric games it is NP-hard to determine whether a given strategy $s_i$ is a mistake~\cite{Gilboa89:Nash,Conitzer08:New}. However, for two-player zero-sum games, this can be determined in polynomial time. 

\begin{definition}
A strategy is a \emph{mistake} if it is played with probability zero in all Nash equilibria.
\end{definition}

Procedure to compute whether strategy $s'_i$ is a mistake: 
\begin{enumerate}
\item Compute an equilibrium by solving the LP; this determines the value of the game to player $i$, $v^*_i.$
\item Solve the LP that maximizes the component of vector variable $x$ corresponding to strategy $s'_i$ from the LP in the preceding step subject to an additional constraint that player $i$'s strategy is an equilibrium (i.e., that the expression used as objective in the preceding step is equal to $v^*_i$). Let $\hat{v}$ denote the optimal objective value of this LP.
\item If $\hat{v} = 0$ then $s'_i$ is a mistake, and otherwise it is not.
\end{enumerate}

The actual linear programs are as follows. First we create the formulation assuming that we play the role of player 1 (i.e., that $i$ = 1). Here $A$ is the $m \times n$ payoff matrix, $e$ and $f$ are the scalar $1$, $E$ is $1 \times m$ matrix of all 1's, and $F$ is $1 \times n$ matrix of all 1's. 

\begin{eqnarray*}
\mbox{maximize}_{x,q} && -q^T f \\
\mbox{ subject to } && x^T (-A) - q^T F \leq 0\\
&& x^T E^T = e^T \\
&& x \geq 0 \\
\end{eqnarray*}

We set $v^*_1$ equal to the optimal objective value from solving this LP. 
Suppose that the potential mistake strategy of player 1 corresponds to his $k$'th pure strategy.

\begin{eqnarray*}
\mbox{maximize}_{x,q} &&x_k \\
\mbox{ subject to } && x^T (-A) - q^T F \leq 0\\
&& x^T E^T = e^T \\
&& x \geq 0 \\
&& -q^T f = v^*_1 \\
\end{eqnarray*}

Then we set $\hat{v}_1$ equal to the optimal objective value from this LP and test whether it is strictly greater than 0.

This procedure can be applied analogously for two-player zero-sum extensive-form games using the sequence-form LP formulation~\cite{Koller96:Efficient}. (That formulation is similar except that the matrices $E$ and $F$ are defined differently, as are $e$ and $f$ which are now vectors and not scalars.)
Furthermore, it can be modified to determine whether action $a_i$ in an extensive-form game is a mistake by setting the objective in the second step to be equal to the action sequence in the LP culminating in $a_i$ for player $i$. (In the sequence-form LP the variables correspond to all possible sequences of actions for the players, plus variables for the game value.)  

Note that this procedure involves actually computing a Nash equilibrium as a subroutine (as the first step), and so we are not viewing the concept as being useful from the perspective of being used as a pre-processing step to reduce the size of the game in advance of equilibrium computation, as several prior concepts considered previously have done. Instead, we are viewing this more at conceptual face value in order to describe actions that are ``bad'' according to a clear and rigorous definition. 

Despite the definition being compelling, it turns out that it is not strong enough to rule out even weakly dominated strategies (though it is clear that every strictly dominated strategy is a mistake). Note that we do not really mind that all weakly dominated strategies are not ruled out, as we do not expect weakly dominated strategies and actions to play any significant role in large realistic games. However, we can strengthen the concept to rule these out as well. 

\begin{proposition}
There exist games that contain a weakly dominated strategy that is not a mistake.
\end{proposition}
\begin{proof}
Consider the game in Figure~\ref{fi:RPST}. This game is similar to Rock-Paper-Scissors, except that a fourth pure strategy labeled Q is added for both players. Q is weakly dominated for player 1 (and Q is then iteratively weakly dominated for player 2 after removing Q for player 1). However, Q is played with positive probability in one of the equilibria: in one P1 randomizes equally between R, P, and S and in the other he randomizes equally between P, S, and Q (while P2 randomizes equally between R, P, and S in both). 
\end{proof}

\begin{figure}[!ht]
\begin{center}
\begin{tabular}{c|*{4}{c|}}
&R &P &S &Q\\ \hline
R &0 &-1 &1 &0\\ \hline
P &1 &0 &-1 &1\\ \hline
S &-1 &1 &0 &0\\ \hline
Q &0 &-1 &1 &-1\\ \hline
\end{tabular}
\caption{Payoff matrix of RPSQ.}
\label{fi:RPST}
\end{center}
\end{figure}

\begin{definition}
A strategy is a \emph{strong mistake} if it is played with probability zero in all non-weakly-dominated Nash equilibria.
\end{definition}
\normalsize

Procedure to compute a strong mistake:
\begin{enumerate}
\item Compute an equilibrium by solving the LP; this determines the value of the game to player $i$, $v_i.$
\item Compute the best response of player $i$ to a fully mixed strategy $s^*_{-i}$ of player $-i$ subject to the constraint that an expected payoff of at least $v_i$ is attained in the worst case. This can be accomplished by solving a second linear program. Let $v'$ denote the optimal objective value. 
\item Solve the LP that maximizes the variable $x$ corresponding to strategy $s'_i$ from the LP in the preceding step subject to additional constraints that player $i$'s strategy is an equilibrium (i.e., that the expression used as objective in the first step is equal to $v_i$) and that player $i$'s strategy achieves the best response value against $s^*_{-i}$ (i.e., that the expression used as the objective in the second step is equal to $v'$). Let $\hat{v}$ denote the optimal objective value of this LP.
\item If $\hat{v} = 0$ then $s'_i$ is a mistake, and otherwise it is not.
\end{enumerate}

This procedure can be applied by solving a similar sequence of linear programs to the standard mistake formulation, using, for example, a uniform random mixture over all pure strategies of the opponent for the strategy $s^*_{-i}$.

The correctness of this procedure follows from the fact that any strategy that is a best response to a fully mixed strategy of the opponent is undominated~\cite{vanDamme87:Stability}. It is clear by definition that a strong mistake now rules out all weakly dominated strategies, though it may not eliminate all iteratively weakly-dominated strategies. We would need to be in quite a specific case for this to occur however. (The strong mistake $s'$ would have to be iteratively weakly dominated by another strategy $s''$ where $s'$ and $s''$ perform identically against a set of opponent strategies that are played with nonzero probability in equilibrium, but $s'$ performs strictly worse against some strategy for the opponent $t'$ and also strictly better against another strategy for the opponent $t''$, where both $t'$ and $t''$ are (iteratively) dominated themselves, where $t''$ can be eliminated first according to a sequence of iterated removal of weakly dominated strategies.) Note that for two-player zero-sum games an undominated Nash equilibrium is equivalent to the equilibrium refinement concept of a trembling-hand perfect equilibrium, and this concept has been applied to solving no-limit Texas hold 'em endgames~\cite{Ganzfried13:Improving}.


Though we are not proposing the concept of a mistake to be used as a part of a pre-processing step to speed up equilibrium computation (since the approach for computing whether a strategy is a mistake involves an equilibrium computation as a subroutine), we can show that in theory had all (or some subset of) the mistake strategies been removed from the game, the set of equilibria in the reduced game is a superset of the set of equilibria in the full game. 

\begin{proposition}
Let $G$ be a two-player zero-sum game, let $s_m$ be a mistake strategy for player 1, and let $G_m$ be the identical game to $G$ except with $s_m$ removed. Suppose $\sigma^*$ is a Nash equilibrium in $G$, and let $\sigma^*_m$ be the restricted strategy profile that is identical to $\sigma^*$ except that the component for strategy $s_m$ for player 1 is omitted. Then $\sigma^*_m$ is a Nash equilibrium of $G_m$. 
\end{proposition}

\begin{proof}
Suppose that $\sigma^*$ is a Nash equilibrium of $G$. Clearly player 1 cannot have a profitable deviation from $\sigma^*_m$ in $G_m$, since his strategy space is reduced while player 2's is identical and player 1 had no profitable deviations from $\sigma^*$ in $G$. Similarly, suppose that player 2 now has a profitable deviation in $G_m$ to strategy $t$, which is not in the support of $\sigma^*$. But the expected payoff for player 2 of $t$ against $\sigma^*$ in $G_m$ is identical to the payoff of $t$ against $\sigma^*$ in $G$, because $\sigma^*$ places probability zero on the mistake strategy $s_m$ for player 1. So player 2 cannot have a profitable deviation in $G_m$, and $\sigma^*_m$ is an equilibrium of $G_m$. 
\end{proof}

By induction, iteratively removing any subset of mistake strategies from the game cannot reduce the set of equilibria. So in theory had we removed all mistakes from the game, this would produce a smaller game that still contains all equilibria from the initial game. This applies to extensive-form as well as strategic-form games, and also to removing mistake actions in extensive-form games.

\section{Strategic-form game experiments}
\label{se:exp-sfg}
Our first set of experiments is on strategic-form games. We repeatedly generated two-player zero-sum strategic-form games with payoffs selected uniformly at random in [-1,1]. Table~\ref{ta:exp-sfg} gives results for the number of strategies that are dominated under several different dominance concepts, as well as mistakes. We generated $m \times n$ games with $m = n$ for values 3, 5, 10, 20, 30, and 50. For $m = 3,5$ we averaged over 10,000 trials, for $10,20$ over 1000 trials, and for $30,50$ over 100 trials. The table shows the average values for player 1 per game over the experiments. (Note that while results for player 2 will differ slightly from player 1 due to variance from the sampling, the expected values are identical because the payoffs are selected uniformly at random.) 

\renewcommand{\tabcolsep}{2pt}
\begin{table}[!ht]
\centering
\begin{tabular}{|*{7}{c|}} \hline
&3 &5 &10 &20 &30 &50 \\ \hline
Avg \# SDS &0.6386 &0.5223 &0.079 &0.001 &0.0 &0.0 \\  
Avg \# WDS &0.6386 &0.5223 &0.079 &0.001 &0.0 &0.0 \\
Avg \# iter. SDS &0.936 &0.7721 &0.089 &0.001 &0.0 &0.0\\
Avg \# iter. WDS &0.936 &0.7721 &0.089 &0.001 &0.0 &0.0\\ 
Avg \# mistakes &1.2011 &2.2313 &4.729 &9.742 &14.92 &24.85\\ \hline
\end{tabular}
\caption{Average number of strategies that are dominated under different concepts and mistakes per game for two-player zero-sum games with uniformly-randomly generated payoffs in [-1,1].}
\label{ta:exp-sfg}
\end{table}

We observe that weak dominance does not rule out any additional strategies beyond strict dominance in these games. This is because the payoffs are selected uniformly at random and the probability that two payoff entries are identical (which is a requirement for a strategy to be weakly but not strictly dominated) is zero. So any strategy that is weakly dominated must also be strictly dominated in this game class. Note that even if we considered other game classes that may be more realistic, adding even a tiny amount of noise to the payoffs would produce this same phenomenon. 

A second observation is that for the games with $m = 30$ and 50, zero strategies were dominated according to any of the domination concepts (and this would continue to be the case for games with $m > 50$). This indicates that domination is a worthless concept for reasonably-sized games, even for strategic-form games (while it is even less useful in ruling out dominated actions in extensive-form games). Domination only plays a useful role for very small games. (Note though that it is possible that the concepts are useful in certain very specific larger games.) 

For $m = 10$ we can start to see domination play a small role, with around 0.08 (less than 0.1\% of the strategies) being dominated, and very slightly more being ruled out from iterative domination. Once we get down to $m = 5$ we see that 0.5223 strategies are dominated (which is slightly over 10\% of the pure strategies), with 0.7721 (slightly over 15\%) iteratively dominated. And for tiny $3 \times 3$ games 0.6383 strategies are dominated on average (slightly over 20\% of the strategies), with 0.936 (slightly over 30\%) iteratively dominated.  

By contrast, the number of mistakes appears to approach 50\% of the total number of strategies for large values of $m$. Thus, in large realistic scenarios there are likely many strategies that are clearly bad and should never be played, and the standard concepts of dominance do not rule any of them out. One can view the number of mistakes as an upper bound on what we can hope to rule out by a concept extending domination that can be used as a potential pre-processing step (as described previously we cannot simply prune mistakes ourselves as a pre-processing step for equilibrium computation since determining whether a strategy is itself a mistake involves computing a Nash equilibrium as a subroutine).

While the disparity is most pronounced for large games, even for smaller $m$ the number of iteratively dominated strategies is far smaller than the number of mistakes.

\section{Extensive-form game experiments}
\label{se:exp-efg}
We next study an extensive-form game that is a generalization of 3-card Kuhn poker previously considered. The rules are identical to Kuhn poker except that the deck consists of $n$ cards from 1--$n$. This game 
has been studied previously in the context of leveraging qualitative models of equilibrium structure to lead to faster algorithms for equilibrium computation~\cite{Ganzfried10:Computing}.

For the $n=3$ version previously considered there were several dominated (and iteratively dominated) actions, that were demonstrated to play a pivotal role in the game. However, as $n$ is larger the percentage of actions that are weakly dominated becomes tiny (as a fraction of the total number of actions), and the role of domination is meaningless. We will show, however, that many of the game's actions are actually mistakes and are not played in equilibrium, and therefore there is again a big disparity between the set of strategies ruled out by dominance and the set of actions that are irrational.

As stated earlier, the only (iteratively)-dominated actions in this game are for player 2 to call vs. a bet with 1, for player 1 to call vs. a bet with 1 after player 1 checks and player 2 bets, for player 2 to fold vs. a bet with card $n$, for player 1 to fold vs. a bet with card $n$ if P1 checks and P2 bets, and for P2 to check with card $n$ after player 1 checks. So 5 total actions are dominated regardless of $n$ (assuming $n > 3$), and no further actions are iteratively dominated.

For $n = 4$, there are 16 different action (sequences) for each player (with 2 being dominated for player 1 and 3 being dominated for player 2). For this game it turns out that there are 5 mistake actions for player 1 and 4 mistake actions for player 2. For $n = 10$, there are 11 mistakes for player 1 and 14 for player 2 (out of 40 total actions). Values for $n$ up to 100 are given in Table~\ref{ta:exp-efg}.

\renewcommand{\tabcolsep}{3pt}
\begin{table}[!ht]
\centering
\begin{tabular}{|*{8}{c|}} \hline
&4 &5 &10 &20 &30 &50 &100 \\ \hline
Num dominated actions P1 &2 &2 &2 &2 &2 &2 &2 \\  
Num dominated actions P2 &3 &3 &3 &3 &3 &3 &3\\
Num mistakes P1 &5 &5 &11 &21 &29 &49 &93\\ 
Num mistakes P2 &4 &7 &14 &31 &46 &64 &93\\
Total num actions P1 &16 &20 &40 &80 &120 &200 &400\\ 
Total num actions P2 &16 &20 &40 &80 &120 &200 &400\\ \hline
\end{tabular}
\caption{Number of dominated actions and mistakes for Generalized Kuhn Poker with different deck size $n$.}
\label{ta:exp-efg}
\end{table}

While the concept of dominated actions seemed promising based on the analysis of the Kuhn poker tournament~\cite{Gibson14:Regret}, the results demonstrate that it is useless in larger realistic games (it would be similarly useless on variants commonly played such as Texas hold 'em). On the other hand, a very large number of actions (nearly $\frac{1}{4}$) are mistakes and should not be taken by rational players.

\section{Conclusion}
\label{se:con}
We presented a new concept of a mistake in two-player zero-sum games, described the relationship to previously-studied concepts, presented an algorithm for computing it, and demonstrated it experimentally in strategic-form and extensive-form games. The concept can be naturally extended to non-zero-sum and multiplayer games, though it is more challenging to compute (since computing a Nash equilibrium is PPAD-hard for those games), and furthermore Nash equilibrium is conceptually a more questionable solution concept for those games. While the concept is not useful in terms of being a pre-processing step for Nash equilibrium, it can be useful from the perspective of evaluating the weaknesses of strategies of bounded rational agents, and in particular for human strategies in large imperfect-information games, which likely have many mistake actions but very few dominated actions. This could have several potential applications, for example, identifying whether human players are illegally utilizing game-theoretic solvers for real-time play in games like online poker. Future work can look into creating a more refined concept of poor actions for extensive-form games that falls between dominated actions and mistakes that is not so strict so that it is common in practice, but is not so general that it can be computed more efficiently than a full equilibrium computation and can be applied as an efficient pre-processing step.

\bibliographystyle{plain}
\bibliography{C://FromBackup/Research/refs/dairefs}
\end{document}